\documentclass[lettersize,journal]{IEEEtran}
\usepackage{amsmath,amsfonts,amssymb}
\usepackage{algorithmic,algorithm}
\usepackage{array}
\usepackage[caption=false,font=normalsize,labelfont=sf,textfont=sf]{subfig}
\usepackage{textcomp}
\usepackage{stfloats}
\usepackage{url}
\usepackage{verbatim}
\usepackage{graphicx}
\usepackage{epsfig}
\usepackage[usestackEOL]{stackengine}
\usepackage{multirow}

\def\BibTeX{{\rm B\kern-.05em{\sc i\kern-.025em b}\kern-.08em
    T\kern-.1667em\lower.7ex\hbox{E}\kern-.125emX}}
\usepackage{balance}
\newtheorem{pro}{Proposition}[section]
\newtheorem{thm}[pro]{\bf Theorem}

\newtheorem{rem}[pro]{\bf Remark}
\newtheorem{ex}[pro]{\bf Example}
\newtheorem{definition}{Definition}
\newenvironment{proof}{{\bf Proof:}}{\hfill$\square$}

\begin{document}
\title{4-Cycle Free Spatially Coupled LDPC Codes with an Explicit Construction}
\author{Author 1, Author 2, Author 3}
\author{Zeinab Roostaie\thanks{Z. Roostaie is a Ph. D student with the Department of Mathematical Sciences, Shahrekord University, P. O. Box 8818634141, Shahrekord, IRAN}, Mohammad Gholami\thanks{Corresponding author: M. Gholami is a Professor in Mathematics with the Department of Mathematical Sciences, Shahrekord University, P. O. Box 8818634141, Shahrekord, IRAN}, Farzad Parvaresh\thanks{F. Parvaresh is an Associate Professor in Department of Electrical Engineering University of Isfahan, E-mail: f.parvaresh@eng.ui.ac.ir}}
%

\maketitle

\begin{abstract}
Spatially coupled low-density parity-check (SC-LDPC) codes are a class of capacity approaching LDPC codes
with low message recovery latency when a sliding window decoding is used. In this paper, we first present a new method for the construction of a class of SC-LDPC codes by the incidence matrices of a given non-negative integer matrix $E$, and then the relationship of 4-cycles between matrix $E$ and the corresponding SC-LDPC code are investigated. Finally, by defining a new class of integer finite sequences, called {\it good sequences}, for the first time, we give an explicit method for the construction of a class of 4-cycle free SC-LDPC codes that can achieve (in most cases) the minimum coupling width.
\end{abstract}


\section{Introduction}

\IEEEPARstart{A}as originally invented by Gallager in 1962, low-density parity-check (LDPC) codes is a type of error correction code with long block lengths which can be decoded efficiently using the low-complexity iterative belief-propagation (BP) decoding algorithms.
In general, LDPC codes consist of two categories: random/ pseudo-random LDPC codes \cite{S. Lin}, and the structured LDPC codes \cite{S. Lin}-\cite{G. Mitchell}. In the ensemble of structured LDPC codes, {\it cyclic} or {\it quasi-cyclic} (QC) LDPC codes have been widely used in many practical applications as they are much easier for hardware implementation~\cite{S. Lin}.


{\it Spatially coupled} (SC) LDPC codes \cite{origin}-\cite{S.Kudekar} can be viewed as a sequence of LDPC block codes whose graph representations are coupled together over time, resulting in a convolutional structure with block-to-block memory. A remarkable property of SC-LDPC codes, established numerically in \cite{Lentmaier} and analytically in \cite{Kudekar}, \cite{S.Kudekar}, is that asymptotically their iterative message passing decoding threshold is equal to the MAP decoding threshold of the underlying LDPC block code ensemble under certain conditions, a phenomenon known as {\it threshold saturation}. In other words, the (exponential complexity) MAP decoding performance of the underlying block code can be achieved by its coupled version with (linear complexity) message passing decoding.


One of the advantages of SC-LDPC codes is that in the coupling process a huge number of absorbing sets in the Tanner graph (TG) are broken \cite{D. G. Mitchell}. This is the main reason for the remarkable performance in waterfall and error floor regions of their bit error rate curves. Moreover, the memorized codeword structure and long coupling length make SC-LDPC more sensitive to short cycles. Small cycles and graphical structures like trapping sets and absorbing sets, which often contain short cycles, are known to have a strong influence on the error floor \cite{Karimi}. Constructing SC-LDPC convolutional codes (SC-LDPC-CCs) with the lowest constraint length and free of short cycles such as 4-cycles is a challenging problem.

In \cite{Mitchell}, a protograph-based method is used to construct some SC-LDPC codes with good performances, using first an {\it edge-spreading} procedure to couple together a chain of block code protographs, followed
by a {\it graph lifting} using permutation matrices. The decoding latency depends on the product of the {\it coupling width} and the {\it graphs lifting factor}, so codes with a small coupling width may result to have a small latency. In \cite{origin}, to find 4-cycle free SC-LDPC codes, a multi-stage design has been proposed based on a heuristic search and a check-and-flip process for the edge-spreading procedure by decomposing the fully-one base matrix $B$ into $w+1$ 4-cycle free component matrices. Moreover, they have employed an exhaustive search to find the minimum $w$, meeting the lower bound of the coupling widths.

In this paper, for a given matrix $E$ of non-negative integers, we deal with two constructions of SC-LDPC codes based on the incidence matrices of the elements of $E$, dependent on whether the indices of the considered incidence matrices are consecutive or not. Then, we pursue the existence of 4-cycles in the PCM of the constructed codes by some square sub-matrices in the primary matrix $E$ or its {\it representative block matrix}. The first representation inherits an explicit construction of 4-cycle free matrix $E$ based on newly defined sequences, called {\it good sequences} to generate a class of 4-cycle free SC-LDPC codes having minimum coupling widths (in most cases).


\section{The Preliminaries and Design Approaches}
Let $B = \left(b_{i,j}\right)_{p\times q}$, be a non-negative integer matrix with Tanner graph ${\rm TG}(B)$ containing {\it check nodes} $\{c_1,\cdots,c_p\}$ and {\it variable nodes} $\{v_1,\cdots,v_q\}$ associated to the rows and columns of $B$, respectively, for which $b_{i,j}$ determines the number of edges between $c_i$ and $v_j$.  Now, associated to $B$ with ${\cal G}={\rm TG}(B)$, the {\it spatially coupled (SC) LDPC} code with {\it base matrix} $B$ can be described as an {\it edge-spreading} procedure to couple together a chain of ${\cal G}$.
\begin{figure}[!t]
\centering
\includegraphics[width=2.5in]{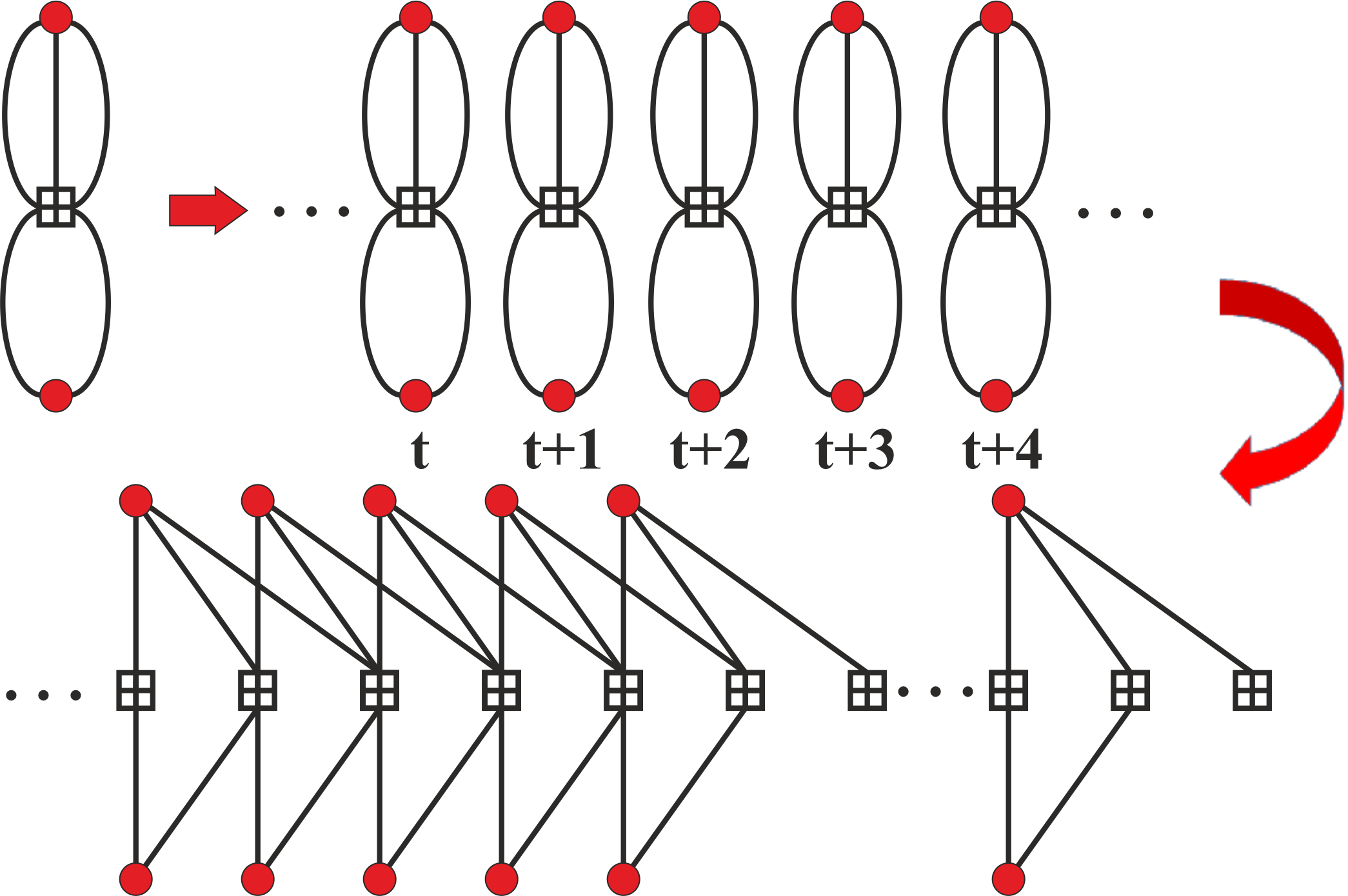}
\caption{Coupling process of the protograph associated to the base matrix $B=\left(3\quad2\right)$ }
\label{fig1}
\end{figure}
 %
In fact, the Tanner graph of the target code can be constructed by first replicating $\cal G$ as an infinite chain, and then spreading edges from the variable nodes of each copy of $\cal G$ by connecting them to the check nodes at most $w$ copies of $\cal G$ after it, i.e $t$, $t+1$, $\cdots$, $t+w$, where $w$ is called {\it coupling width}. For example, Fig.~\ref{fig1} shows the process of constructing the Tanner graph of an SC-LDPC code with coupling width $w=2$, by coupling TG$(B)$, in which $B = \left(3\quad 2\right)$.

It can easily be seen that the above coupling process to construct an SC-LDPC code with the base matrix $B$ corresponds to the decomposition of $B$ into $w + 1$ {\it component matrices} $B_0, B_1, \cdots, B_w$ of the same size, such that $B=\sum_{i=0}^{w}B_i$, therefore this construction is not unique in general, and strongly depends on the corresponding decomposition. For example, in Fig.~\ref{fig1}, we have $B=B_0+B_1+B_2$, in which $B_0=B_1=\left(1\quad 1\right)$, and $B_2=\left(1\quad0\right)$.

For the base matrix $B$ with decomposition $B=\sum_{i=0}^{w}B_i$, it can be seen that the PCM of the corresponding SC-LDPC code is as follows:
\begin{equation}\label{eq0}
{\cal H}(B_0,\cdots,B_w)=\small\left(\begin{array}{ccccc}
                    B_{0} &  &  \\
                    B_1 & B_0 &  \\
                    \vdots & \vdots & \ddots \\
                    B_w & B_{w-1} & \ddots \\
                     & B_w &\ddots  \\
                     &  &\ddots \\
                      &  & \\
                  \end{array}\right)
\end{equation}
In practice, an SC-LDPC code is terminated after a finite number $L$ copies of $\cal G$, where $L$ is called the {\it coupling length}. Then, the following matrix ${\cal H}^{(L)}$ can be considered as the PCM of the SC-LDPC code with coupling length $L$.
$$
{\cal H}^{(L)}(B_0,\cdots,B_w)=\small\left(\begin{array}{ccccc}
                    B_{0} &  &  \\
                    B_1 & B_0 &  \\
                    \vdots & \vdots & \ddots \\
                    B_w & B_{w-1} & \ddots &B_0\\
                     & B_w &\ddots &B_1 \\
                     &  &\ddots &\vdots \\
                      &  & &B_w\\
                  \end{array}\right)_{(L+w)p\times Lq}
$$

It is important to note that the first and last $w p$ check nodes have a reduced degree, that is, terminating the base graphs causes a small amount of disorder at both ends of the graph, which is a special feature that affects the saturation threshold of SC-LDPC codes ~\cite{Mitchell}, \cite{Lentmaier}-\cite{S.Kudekar}.


The constraint length and rate of the SC-LDPC code with PCM ${\cal H}^{(L)}$ is $v=p(w+1)$ and
$R^{(L)}=1-\frac{(L+w)p}{Lq}$, respectively, and the approximate rate is given by ${\lim}_{L\to\infty}R^{(L)}\triangleq R^{\infty}=1-\frac{p}{q}$.
$H^{(L)}$  defines a particular SC-LDPC code, whose girth (denoted by $g({\cal H}^{(L)})$) is the length of the shortest cycle in the corresponding Tanner graph. 




\section{A Deterministic Approach to Construct a Class of SC-LDPC Codes}
 For positive integers $p,q$, $p<q$, let $E=(e_{i,j})_{p\times q}$ be a matrix with entries belonging to the set of non-negative integers and ${\cal E}=\{e_{i,j}: 1\le i\le p, 1\le j\le q\}$ be the {\it set of all elements (SOE)} of the matrix $E$. Associated to each element $e\in{\cal E}$, the {\it incidence matrix} ${\cal M}_e$ can be defined as the $p\times q$ binary matrix ${\cal M}_e$ whose $(i,j)$the entry ${\cal M}_e(i,j)$ is non-zero, if and only if $e_{i,j}=e$. For the case that $e\not\in{\cal E}$, we set ${\cal M}_e$ to be the $p\times q$ zero matrix.
For example, associated with the matrix
\begin{equation}\label{eq2}
E=\left(\begin{array}{cccc}
3 & 0 & 1 & 3
\\
 4 & 3 & 3 & 0
\\
 4 & 0 & 5 & 5
\end{array}\right)
\end{equation}
we have the following incidence matrices:
$$\begin{array}{l}
{\cal M}_0=\left(\begin{array}{cccc} 0100\\ 0001\\ 0100\end{array}\right),
{\cal M}_1=\left(\begin{array}{cccc} 0010\\ 0000\\ 0000\end{array}\right),
{\cal M}_2=\left(\begin{array}{cccc} 0000\\ 0000\\ 0000\end{array}\right),\vspace{1mm}\\
{\cal M}_3=\left(\begin{array}{cccc} 1001\\ 0110\\ 0000\end{array}\right),
{\cal M}_4=\left(\begin{array}{cccc} 0000\\ 1000\\ 1000\end{array}\right),
{\cal M}_5=\left(\begin{array}{cccc} 0000\\ 0000\\ 0011\end{array}\right)
\end{array}
$$
\begin{definition}
Matrix $E$ is called {\it 4-cycle free} if
$e_{i_2,j_1}-e_{i_1,j_1}\neq e_{i_1,j_2}-e_{i_2,j_2}$, for each $1\le i_1<i_2\le p$ and $1\le j_1<j_2\le q$.
\end{definition}

 In fact, each 4-cycle free matrix $E$ can be considered as the exponent matrix of a QC-LDPC code with circulant permutation matrix (CPM) size $N$, where $N$ is large enough~\cite{M. Karimi}.

Associated to a matrix $E=(e_{i,j})_{p\times q}$, with SEO ${\cal E}$,
we follow a structure for the construction of PCM of a class of SC-LDPC codes.

{\bf Design Method.} Let $I=\{i_0,i_1,\cdots,i_w\}$, $i_0<i_1<\cdots<i_w$, be a set of non-negative integers containing $\cal E$. For the fully-1 $p\times q$ base matrix $B$, we follow the design approach by breaking $B$ to the $|I|=w+1$ component matrices ${\cal M}_{i}$, $i\in I$, such that $B=\sum_{i\in I}{\cal M}_{i}$. In fact, associated with $E$, we consider the incidence matrices $\{{\cal M}_i, i\in I\}=\{{\cal M}_{i_0},{\cal M}_{i_1},\cdots,{\cal M}_{i_w}\}$ as the target component matrices to define the PCM of the corresponding $(p,q)$ SC-LDPC code to be ${\cal H}(E)={\cal H}({\cal M}_{i_0},{\cal M}_{i_1},\cdots,{\cal M}_{i_w})$, as defined in Eq.~\ref{eq0}.


For positive integer $L$, ${\cal H}^{(L)}(E)={\cal H}^{(L)}({\cal M}_{i_0},{\cal M}_{i_1},\cdots,{\cal M}_{i_w})$ can be constructed from ${\cal H}(E)$ by terminating the matrix after $L$ steps which contain $Lq$ (first) variable nodes and $(L+i_w-i_0)p$ (first) check nodes, starting from the upper-left corner of ${\cal H}(E)$. The design rate of the code with PCM ${\cal H}^{(L)}(E)$ is at least $1-\frac{(L+i_w-i_0)p}{Lq}$ which tends to $1-\frac{p}{q}$, when $L$ enlarges.

Besides the rate, we need to analyze some other properties influencing the error-rate performances, such as the girth. Before, for integers $a,b$, $a<b$, we set $[a,b]$ to denote the set of all integers between $a$ and $b$, i.e. $[a,b]=\{a,a+1,\cdots,b\}$.
\begin{thm}
\label{thm0}
If $I=[i_0,i_w]$, then the SC-LDPC code with PCM ${\cal H}(E)$ is 4-cycle free if and only if $E$ is 4-cycle free.
\end{thm}
\begin{proof}
The existence of a 4-cycle in $E=(e_{i,j})_{p\times q}$ corresponds to a $2\times 2$ sub-matrix
${\cal A}=\left(\begin{array}{cc} e_{i_1,j_1} & e_{i_1,j_2} \\ e_{i_2,j_1} & e_{i_2,j_2} \end{array}\right)$ of $E$, in which $e_{i_1,j_1}-e_{i_1,j_2}+e_{i_2,j_2}-e_{i_2,j_1}=0$. \\
For this case, one of the following cases may have occurred:

{Case 1.} ${\cal A}=\left(\begin{array}{cc} a & a \\ a& a \end{array}\right)$, for some $a\in{\cal E}$. In this case, we have a 4-cycle in the incidence matrix ${\cal M}_a$, because ${\cal M}_a(i_1,j_1)={\cal M}_a(i_1,j_2)={\cal M}_a(i_2,j_1)={\cal M}_a(i_2,j_2)=1$, which induces a 4-cycle in the Tanner graph of ${\cal H}_1(E)$.

{Case 2.} If ${\cal A}$ has one the forms $\left(\begin{array}{cc} a & a \\ b & b \end{array}\right)$ or
$\left(\begin{array}{cc} a & b \\ a & b \end{array}\right)$, $a,b\in{\cal E}$, then we have 4-cycles in $\left(\begin{array}{c}{\cal M}_a \\ {\cal M}_b\end{array}\right)$ or
$\left(\begin{array}{cc}{\cal M}_a & {\cal M}_b\end{array}\right)$, respectively. Because, for the first case, we have ${\cal M}_a(i_1,j_1)={\cal M}_a(i_1,j_2)=1$ and ${\cal M}_b(i_2,j_1)={\cal M}_b(i_2,j_2)=1$ and for the second case, we have ${\cal M}_a(i_1,j_1)={\cal M}_a(i_2,j_1)=1$ and ${\cal M}_b(i_1,j_2)={\cal M}_b(i_2,j_2)=1$.

{Case 3.} The elements of $\cal A$ are different. In this case, without loss of generality, we assume that ${\cal A}=\left(\begin{array}{cc}a&b\\c&d\end{array}\right)$, with $b<a<d<c$ and $c-a=d-b$. In this case, we have a 4-cycle in $\left(\begin{array}{cc}{\cal M}_{a} & {\cal M}_b\\ {\cal M}_{c}&{\cal M}_{d}\end{array}\right)=\left(\begin{array}{cc}{\cal M}_{b+h} & {\cal M}_b\\ {\cal M}_{b+h+h'}&{\cal M}_{b+h'}\end{array}\right)$ which clearly is a sub-matrix of ${\cal H}_1(E)$, where $h=a-b=c-d$ and $h'=d-b=c-a$. The 4-cycle is because of ${\cal M}_a(i_1,j_1)={\cal M}_b(i_1,j_2)={\cal M}_c(i_2,j_1)={\cal M}_d(i_2,j_2)=1$.
\end{proof}
\begin{ex}
The matrix $E$ given by Eq.~\ref{eq2} is 4-cycle free and it can be considered as the exponent matrix of a QC-LDPC code with girth 6 and CPM-size $N=7$. Then, we have the following PCM of an SC-LDPC code with girth 6.
$$\small\begin{array}{l}
{\cal H}_1(E)=\\\\
\left(
                  \begin{array}{ccc}
                    {\cal M}_0 &  &  \\
                    {\cal M}_1 & {\cal M}_0 &  \\
                    {\cal M}_2 & {\cal M}_1 &  \\
                    {\cal M}_3 & {\cal M}_2 & \ddots \\
                    {\cal M}_4 & {\cal M}_3 & \ddots \\
                    {\cal M}_5 & {\cal M}_4 & \ddots \\
                         & {\cal M}_5 & \ddots \\
                  \end{array}
\right)=\end{array}
{\scriptsize\left(
                  \begin{array}{lll}
                    \begin{array}{p{.1mm}p{.1mm}p{.1mm}p{.1mm}} .&1&.&.\\ .&.&.&1\\ .&1&.&.\end{array}\\
                    \begin{array}{p{.1mm}p{.1mm}p{.1mm}p{.1mm}} .&.&1&.\\ .&.&.&.\\ .&.&.&.\end{array}&
                    \begin{array}{p{.1mm}p{.1mm}p{.1mm}p{.1mm}} .&1&.&.\\ .&.&.&1\\ .&1&.&.\end{array}\\
                    \begin{array}{p{.1mm}p{.1mm}p{.1mm}p{.1mm}} .&.&.&.\\ .&.&.&.\\ .&.&.&.\end{array}&
                    \begin{array}{p{.1mm}p{.1mm}p{.1mm}p{.1mm}} .&.&1&.\\ .&.&.&.\\ .&.&.&.\end{array}\\
                    \begin{array}{p{.1mm}p{.1mm}p{.1mm}p{.1mm}} 1&.&.&1\\ .&1&1&.\\ .&.&.&.\end{array}&
                    \begin{array}{p{.1mm}p{.1mm}p{.1mm}p{.1mm}} .&.&.&.\\ .&.&.&.\\ .&.&.&.\end{array}& \ddots \\
                    \begin{array}{p{.1mm}p{.1mm}p{.1mm}p{.1mm}} .&.&.&.\\ 1&.&.&.\\ 1&.&.&.\end{array}&
                    \begin{array}{p{.1mm}p{.1mm}p{.1mm}p{.1mm}} 1&.&.&1\\ .&1&1&.\\ .&.&.&.\end{array}& \ddots \\
                    \begin{array}{p{.1mm}p{.1mm}p{.1mm}p{.1mm}} .&.&.&.\\ .&.&.&.\\ .&.&1&1\end{array}&
                    \begin{array}{p{.1mm}p{.1mm}p{.1mm}p{.1mm}} .&.&.&.\\ 1&.&.&.\\ 1&.&.&.\end{array} & \ddots \\
                    &\begin{array}{p{.1mm}p{.1mm}p{.1mm}p{.1mm}} .&.&.&.\\ .&.&.&.\\ .&.&1&1\end{array} & \ddots \\
                  \end{array}
                \right)}
$$
\end{ex}

In Theorem~\ref{thm0}, if $I=[a,b]$ is the least interval containing the SOE of $E$, then ${\cal H}(E)$ can be used as the PCM of a 4-cycle free SC-LDPC code with coupling width $w=b-a$. Moreover, because of the direct role of matrices having property 4-cycle free in constructing SC-LDPC codes with girth 6, it is enough to focus our attention on the construction of these primary matrices. For example, some 4-cycle free matrices in \cite{M. Karimi} have been considered as the exponent matrices of some QC-LDPC codes with girth 6. In fact, the authors have proposed some 4-cycle free matrices $E=(e_{i,j})_{p\times q}$ with $e_{i,j}=(i-1)(j-1)\bmod n$, in which $n$ is the smallest prime greater than $q$. Then, for some $p,q$, $2\le p\le 5$ and $p<q\le 14$, Table~\ref{tab1} has provided the coupling widths of the codes in \cite{M. Karimi} in a column labeled with $w_2$.

Applying the above approach to generate $(p,q)$ SC-LDPC codes with girth 6, it is desirable to identify what is the minimum $w$ required to achieve girth 6, since we will typically want small $w$ to minimize latency when using sliding window (SW) decoding. As shown in Table~\ref{tab1}, for some different $p,q$, the authors in \cite{origin} have tried to find the minimum possible $w$ by breaking $B$ to $w+1$ {\it component matrices} such that the {\it excluded patterns} in the PCM of the final SC-LDPC code are 4-cycle free. The approach in \cite{origin}, however, is not efficient, because the number of the {excluded patterns} enlarges significantly by increasing the size of the base matrix $B$. Moreover, the {component matrices} are constructed {randomly} which are modified by a {\it check-and-flip} process (when component matrices or excluded patterns contain 4-cycles) which imposes an additional complexity to the algorithm for construction of the desired component matrices.

Here, for the first time, we propose an explicit method to generate the component matrices such that the final SC-LDPC codes are 4-cycle free. In fact, the component matrices are constructed by some new well-defined {\it good sequences} which are generated by a {\it depth first} algorithm. The outputs show that the final SC-LDPC codes in most cases match the minimum bound latency $w$ reported in Lemma 4 of \cite{lowbound}.
\begin{definition}
\label{def1}
Let $p,q,l$ be positive integers with $p<q$. A finite non-negative integer sequence $(a_n^{(p, q)})_{1\le n\le p+q-1}$ is called a {\it good sequence}, if the elements are satisfied in the following condition:
\begin{align}\label{eq3}
  a_{n_1-n_2+n_3}^{(p, q)}+a_{n_1}^{(p, q)}-a_{n_2}^{(p, q)}+a_{n_3}^{(p, q)}\neq0
\end{align}
 for each $1\le n_1<n_2<n_3\le p+q-1$, $n_2-n_1<p$ and $n_3-n_2<q$. In this case, $w=\max\{a_n^{(p, q)}: 1\le n\le p+q-1\}$ is called the maximum of elements (MOE) of the sequence.
\end{definition}

For example, for $(p,q)=(5,10)$, the finite sequence $(a_n^{(5,10)})_{1\le n\le 14}=(6, 5, 3, 6, 1, 1, 7, 3, 0, 2, 7, 6, 7, 0)$ is a good sequence of length $p+q-1=14$ with MOE $w=7$. For given $p,q$, $p<q$ and (enough large) positive integer $w$, Algorithm~\ref{alg1} use a back-tracking depth-first approach to find a good sequence.
In continue, we show that each good sequence with MOE $w$ can be used to generate a 4-cycle free SC-LDPC code with coupling width $w$.
\begin{thm}
\label{thm1}
 For given positive integers $p,q$, $p<q$, let $(a_n^{(p, q)})_{1\le n\le p+q-1}$ be a {\it good sequence} and $E$ be the matrix $(e_{i,j})_{p\times q}$, for which $e_{i,j}=a_{j-i+p}^{(p,q)}$. Then, $E$ is a 4-cycle free matrix.
 \end{thm}
\begin{proof}
 To prove the claim, it is sufficient to show $e_{i_2,j_1}-e_{i_1,j_1}\neq e_{i_1,j_2}-e_{i_2,j_2}$, for each $1\le i_1<i_2\le p$ and $1\le j_1<j_2\le q$. By setting $e_{i,j}=a_{j-i+p}^{(p,q)}$, this is equivalent to show $a_{n_1}^{(p,q)}-a_{n_2}^{(p,q)}\ne a_{n_3}^{(p,q)}-a_{n_4}^{(p,q)}$, for $n_1=j_1-i_2+p$, $n_2=j_1-i_1+p$, $n_3=j_2-i_1+p$ and $n_4=j_2-i_2+p$. However, $1\le n_1<n_2<n_3\le p+q-1$, $n_4=n_1-n_2+n_3$, $n_2-n_1=i_2-i_1<p$ and $n_3-n_2=j_2-j_1<q$. Then, the condition $a_{n_4}^{(p,q)}+a_{n_1}^{(p,q)}- a_{n_2}^{(p,q)}+a_{n_3}^{(p,q)}\neq0$ is hold from Definition~\ref{def1} and the proof is completed.
 \end{proof}
\medskip

In generating good sequences, a challenging problem is finding the sequences with the smallest MOEs, which (according to Theorem~\ref{thm1}) is equivalent to finding SC-LDPC codes with the smallest coupling widths. Then, in Table \ref{tab1}, we have tried to find such good sequences, labeled with $(a_n)$, having the smallest MOE $w_3$. In fact, as the outputs show, our results (in most cases) match well with Lemma 4 of \cite{lowbound}, which gives a lower-bound on $w$ needed to ensure a 4-cycle-free SC LDPC code. On the other hand, MOE's are the same with the minimum coupling width $w_1$ of the codes in \cite{origin}, just for $(p,q)=(5,10)$ ours is better. Besides, because of the high complexity of the approach used in \cite{origin}, there is not any report for $w_1$, especially when $p+q\ge 15$, while the outputs of Algorithm 1 meet the minimum values up to $q=16$ and other values also can be easily obtained even for larger $p$ and $q$ not reported in the table.
\begin{ex}
\label{ex1}
For $(p,q)=(3,6)$, the finite sequence $(0,3,2,0,0,1,3,0)$ is a good sequence with MOE $w=3$. Then, Theorem~\ref{thm1} can be used to construct the following 4-cycle free matrix $E$.
\begin{equation}\label{eq1}
E= \left(\begin{array}{cccccc}
2 & 0 & 0 & 1 & 3 & 0
\\
 3 & 2 & 0 & 0 & 1 & 3
\\
 0 & 3 & 2 & 0 & 0 & 1
\end{array}\right)
\end{equation}
Now, applying Theorem~\ref{thm0}, the matrix ${\cal H}(E)$  is the PCM of a 4-cycle free SC-LDPC code with coupling width $w=3$.
\end{ex}


\begin{algorithm}
\caption{Generating good sequences}\label{alg1}
\begin{algorithmic}
\REQUIRE Let $w$, $p$ and $q$ be given with $p<q$.
\STATE $A\leftarrow\{0,1,\cdots,w\}$, $C_1=\{c_{1,0},\cdots,c_{1,w}\}\leftarrow A$, $n_4\leftarrow 2$
\STATE $a_1^{(p,q)}\leftarrow c_{1,0}$, $l\leftarrow p+q-1$, $q_1\leftarrow True$
\WHILE{$|C_1|>0$ and $n_4\le l$}
 \STATE $B\leftarrow\emptyset$
 \FOR {$n_1$ from 1 to $n_4-2$}
  \FOR {$n_3$ from $n_1+1$ to $n_4-1$}
   \STATE $n_2\leftarrow n_1-n_3+n_4$
   \IF {$n_4-n_2<p$ and $n_4-n_3<q$}
    \STATE $B\leftarrow B\bigcup\{a_{n_2}^{(p,q)}+a_{n_3}^{(p,q)}-a_{n_1}^{(p,q)}\}$
   \ENDIF
  \ENDFOR
 \ENDFOR
 \STATE $C_{n_4}\leftarrow A \setminus B$
 \WHILE {$n_4>0$ and $|C_{n_4}|=0$}
  \STATE $n_4\leftarrow n_4-1$
  \IF{$n_4>0$}
   \STATE $C_{n_4}\leftarrow C_{n_4}\setminus \{a_{n_4}^{(p,q)}\}$
  \ENDIF
 \ENDWHILE
 \IF{$n_4>0$}
  \STATE $v\leftarrow|C_{n_4}|-1$ and $C_{n_4}=\{c_{n_4,0},\cdots,c_{n_4,v}\}$
  \STATE $a_{n_4}\leftarrow c_{n_4,0}$ and $n_4\leftarrow n_4+1$
 \ENDIF
 \IF{$n_4=l+1$}
  \RETURN $(a_n^{(p,q)})_{1\le n\le l}$
 \ENDIF
\ENDWHILE
\end{algorithmic}
\end{algorithm}

\begin{table*}
\centering
\scriptsize
\caption{A comparison between the coupling widths of the constructed codes and the codes in \cite{M. Karimi,origin}\label{tab1}}
\begin{tabular}{cc}
\begin{tabular}{c|c|c|c|c|c}
  p&q& $w_1$ &$w_2$ & $w_{3}$ & $a_n$ \\
  \hline\hline
   &3 & 1 &2 & 1 & (0,0,1,0) \\
   &4 & 2 & 3& 2 & (0,0,1,0,2)\\
   &5 & 2 & 4& 2 & (0,0,1,0,2,0)\\
   &6 & 3 & 5& 3 & (0,1,1,3,2,0,3)\\
   &7 & 3 & 6& 3 & (0,0,1,3,2,0,3,0)\\
  2&8 & 4 & 7& 4 & (0,2,3,2,0,0,4,1,4)\\
   &9 & 4 & 8& 4 & (0,4,4,3,1,2,4,0,3,0)\\
   &10& 5 & 9& 5 & (0,0,1,5,2,4,0,3,2,0,5)\\
   &11& 5 & 10& 5 & (0,5,4,1,3,3,4,0,4,2,5,0)\\
   &12& 6 & 11& 6 & (0,2,0,5,1,1,0,3,4,1,5,0,6)\\
   &13& - &12 & 6 & (0,0,1,0,4,1,6,2,0,3,5,0,6,0)\\
   &14& - & 13& 7 & (0,6,7,1,5,0,7,3,2,4,7,7,4,2,7)\\\hline
   &4 & 2 & 4 & 2 & (0,2,2,0,1,0)\\
   &5 & 2 & 6 & 2 & (0,2,1,2,0,0,2)\\
   &6 & 3 & 6 & 3 & (0,2,3,0,3,1,0,0)\\
   &7 & 3 &10 & 3 & (0,3,1,0,0,2,3,0,3)\\
   &8 & 4 &10 & 4 & (0,1,3,2,0,4,4,0,3,0)\\
  3&9 & 4 & 10& 4 & (0,3,1,2,4,0,4,4,1,0,3)\\
   &10& 5 & 10& 5 & (0,1,4,2,1,5,1,3,0,5,5,0)\\
   &11& 5 & 12& 5 & (0,5,3,5,0,0,4,5,2,1,4,0,5)\\
   &12& 6 & 12& 6 & (0,4,1,0,3,5,6,2,0,6,6,1,6,0)\\
   &13& - & 16& 7 & (0,7,5,1,3,7,7,2,7,1,4,3,0,6,7)\\
   &14& - & 16& 7 & (0,5,1,7,7,2,0,7,4,6,0,3,7,6,7,0)\\\hline
   \end{tabular}&
   \begin{tabular}{c|c|c|c|c|c}
   p&q& $w_1$ &$w_2$ & $w_{3}$ & $a_n$ \\
  \hline\hline
   &5 & 2 & 6 & 2 & (0,2,1,2,0,0,2,1)\\
   &6 & 3 & 6 & 3 & (1,2,0,3,0,0,2,3,1)\\
   &7 & 4 &10 & 4 & (0,2,0,4,4,3,0,1,3,1)\\
   &8 & 5 & 10& 5 & (0,1,5,5,2,0,5,0,3,2,4)\\
  4&9 & 5 & 10& 6 & (0,2,0,4,5,0,6,6,5,1,3,1)\\
   &10& 6 & 10& 6 & (0,2,5,6,4,0,4,3,0,0,6,1,4\\
   &11& - & 12& 6 & (0,0,6,4,3,0,5,1,5,0,3,5,6,0)\\
   &12& - & 12& 7 & (0,0,2,5,0,6,4,1,0,7,3,7,0,0,2)\\
   &13& - & 16& 7 & (0,0,7,6,2,0,6,1,4,1,6,0,2,3,7,0)\\
   &14& - & 16& 8 & (0,6,6,8,2,1,6,4,0,3,7,0,7,8,5,0,0) \\
   &15& - & 16& 8 & (0,0,4,6,7,1,8,3,0,3,8,1,7,6,4,0,0,4) \\
   &16& - & 16& 9 & (0,0,3,8,9,1,0,8,5,0,9,2,6,0,7,9,7,3,3) \\
   \hline
   &6 & 4 & 6 & 4 & (0,3,1,2,4,4,1,4,2,3)\\
   &7 & 4 & 10 & 4 & (0,1,1,0,4,2,4,0,1,1,0)\\
   &8 & 5 & 10 & 6 & (0,5,4,1,3,6,0,0,1,6,5,2)\\
   &9 & 6 & 10 & 6 & (0,2,6,6,4,5,4,0,6,1,3,6,6)\\
  5&10& 8 & 10 & 7 & (0,4,6,0,5,0,7,3,2,0,0,1,4,6)\\
   &11& - & 12 & 7 & (0,3,7,5,6,1,6,5,7,3,0,0,3,7,5)\\
   &12& - & 12 & 8 & (0,0,6,4,8,2,3,5,0,5,8,1,0,0,7,5)\\
   &13& - & 16 & 8 & (0,4,0,1,7,5,8,3,3,8,5,7,1,0,4,0,1)\\
   &14& - & 16 & 9 & (0,4,9,0,6,1,8,9,2,2,0,9,6,9,5,4,8,0)\\
   &15& - & 16 &10 & (0,0,1,3,6,10,2,10,1,0,10,8,5,0,9,2,2,3,5)\\
   &16& - & 16 & 10 & (0,3,10,7,7,5,1,10,1,0,10,2,6,1,7,8,10,0,5,8)\\\hline
\end{tabular}
\end{tabular}
\end{table*}

Although Theorem~\ref{thm0} is useful when $I$ is an interval, we need criteria to check the girth of TG$({\cal H}(E))$ when the elements of $I$ are not consecutive. For this, we first need some definitions from \cite{origin}.
\begin{definition}
Corresponding to the matrix $E$ with SOE $\cal E$, ${\cal E}\subseteq I=\{i_0,i_1,\cdots,i_w\}$ and incidence matrices ${\cal M}_{i_j}$, $0\le j\le w$, the {\it representative block} matrix is defined as follows.
$$B_R(E)=\left(\begin{array}{ccccc}
                    {\cal M}_{i_w} &  {\cal M}_{i_{w-1}} & \ldots& {\cal M}_{i_0} \\
                    &{\cal M}_{i_w} &  \ldots& {\cal M}_{i_1} \\
                    &&\ddots & \vdots \\
                   &&&{\cal M}_{i_w}\\
                  \end{array}\right)$$
\end{definition}

Referring to Lemma 2 in \cite{origin}, it can be seen that ${\cal H}(E)$ is 4-cycle free if and only if $B_R(E)$ is 4-cycle free. Because, each combination of one, two, or four incidence matrices ${\cal M}_{i_j}$, $0\le j\le w$, that contain a 4-cycle in ${\cal H}(E)$ will appear in $B_R(E)$.

Because of the importance of the indices $i\in{\cal E}$ of ${\cal M}_{i}$ in the upper-right of $B_R(E)$, hereinafter we consider $B^{(I)}_R(E)$ to be the following matrix.
\begin{center}
\vspace{-.3cm}$B_R^{(I)}(E)=\left(\begin{array}{ccccc}
                    {i_w} &  {i_{w-1}} & \ldots& {i_0} \\
                    -1&{i_w} &  \ldots& {i_1} \\
                    \vdots&\vdots&\ddots & \vdots \\
                   -1&-1&\ldots&{i_w}\\
                  \end{array}\right)$
\end{center}
in which we consider the indices of the zero elements of $B_R^{(I)}(E)$ to be -1. Here, we give a theorem that connects the existence of 4-cycles in $B_R(E)$ with the existence of some square sub-matrices in $E$.
\begin{thm}
\label{thm2}
 The SC-LDPC code with PCM ${\cal H}(E)$ is 4-cycle free if and only if each $2\times 2$ sub-matrix of the upper-right corner of $B^{(I)}_R(E)$ and each of the square matrices $\left(\begin{array}{cc} a & a \\ b & b \end{array}\right)$ or
$\left(\begin{array}{cc} a & b \\ a & b \end{array}\right)$ $a,b\in{\cal E}$, does not appear in $E$.
 \end{thm}
 \begin{proof}
 Each cycle of length 4 in TG$({\cal H}(E))$ may be occurred in ${\cal M}_{i_j}$, for a $0\le j\le w$, or in the array $\left(\begin{array}{cc}{\cal M}_{i_j} & {\cal M}_{i_k}\end{array}\right)$, or its transpose, for some $0\le j\ne k\le w$, or in the array $\left(\begin{array}{cc}{\cal M}_{i_j} & {\cal M}_{i_k}\\ {\cal M}_{i_{j+h}}&{\cal M}_{i_{k+h}}\end{array}\right)$, for some $0\le k< j\le w$ and $1\le h\le w-j$. However, three first cases are not happen because $E$ does not contain the sub-matrix $\left(\begin{array}{cc} i_j & i_j \\ i_k & i_k \end{array}\right)$ or its transpose, for each $0\le j\ne k\le w$. Besides, the last case may not be considered because $\left(\begin{array}{cc}{i_j} & {i_k}\\ {i_{j+h}}&{i_{k+h}}\end{array}\right)$ is a sub-matrix of $B^{(I)}_R(E)$ which is not included in the matrix $E$.

 \end{proof}
 \begin{rem}
 It is noticed that the matrix $E$ satisfying Theorem \ref{thm2} is not necessary 4-cycle free, although, ${\cal H}(E)$ is 4-cycle free. For example, the following matrix with $I=\{0,2,3,4,5,6,7,8\}$ is a good candidate for Theorem \ref{thm2}, i.e. ${\cal H}(E)$ is 4-cycle free, while $E$ has some 4-cycles.
 $$ E=\left( \begin {array}{ccccccc}
  8&2&5&5&7&3&8\\
  8&0&8&6&0&7&0\\
  4&6&7&3&8&2&3
  \end {array} \right)$$
  For the case that $I$ is an interval containing consecutive elements, Theorem~\ref{thm2} is the same as Theorem~\ref{thm0}.
 \end{rem}
 \begin{figure}
 \centering
\includegraphics[scale=0.55]{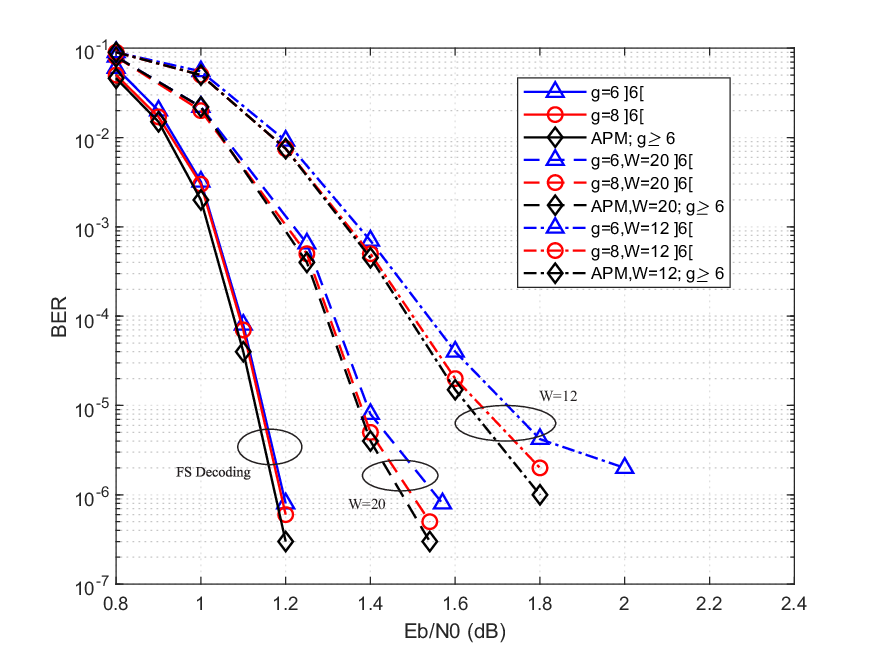}
\caption{BER performance of the constructed codes versus the codes in \cite{origin}\label{fig1}}\vspace{-.5cm}
\end{figure}

\section{APM-LDPC codes}
Let $m$ be a positive integer and $(s,a)\in\mathbb{Z}_m\times \mathbb{Z}_m^*$, for which $\mathbb{Z}_m=\{0,1,\cdots,m-1\}$ and $\mathbb{Z}_m^*=\{a\in\mathbb{Z}_m:\gcd(a,m)=1\}$. By the {\it affine permutation matrix (APM)} ${\cal I}^{s,a}$ with APM-size $m$~\cite{abedi}, we mean the permutation matrix $(e_{i,j})_{0\le i,j\le m-1}$, in which $e_{i,j}=1$ if and only if $i-ja=s\bmod m$. Especially, ${\cal I}^{s,1}$ is a {\it circulant permutation matrix (CPM)} and briefly is denoted by ${\cal I}^s$.
For $E$ in Eq. \ref{eq1} and coupling length $L=100$, ${\cal H}^{(L)}(E)$ can be considered as the base matrix of an APM SC-LDPC code with lifting degree $m=100$, rate $0.485$ and length 60000, after substituting each 0 and 1 element of the base matrix by the $m\times m$ zero matrix and an $m\times m$ APM matrix ${\cal I}^{s,a}$, respectively, such that $(s,a)\in\mathbb{Z}_m\times \mathbb{Z}_m^*$ are selected randomly such that the constructed code is 4-cycle free.

Finally, we have used the Sliding window (SW) decoding with different {\it window size} $W\in\{12,20\}$ and Standard flooding schedule (FS) decoding in Figure~\ref{fig1} to show the bit-error-rate performance of the constructed 4-cycle free APM-LDPC code against the QC SC-LDPC codes in~\cite{origin} with girths 6 and 8, where the maximum-iteration is 100. As the figure shows, the constructed APM SC-LDPC code outperform slightly the QC SC-LDPC codes in~\cite{origin}.




\end{document}